\documentclass[12pt]{amsart}
\usepackage[english]{babel}
\usepackage{amsmath,amsthm,amsfonts,amssymb,epsfig,color,ulem,vector,hyperref}
\usepackage[left=1in,top=1in,right=1in]{geometry}

\newtheorem{thm}{Theorem}[section]
\newtheorem{lem}[thm]{Lemma}

\newtheorem{df}[thm]{Definition}

\newtheorem{ass}[thm]{Assumption}

\newcommand{\R}{\mathbb{R}}
\newcommand{\Z}{\mathbb{Z}}

\newcommand{\F}{\mathcal{F}}

\begin{document}

\title{Zero-Temperature Fluctuations in Short-Range Spin Glasses}

\author[L.-P. Arguin]{L.-P. Arguin}            
 \address{L.-P. Arguin\\ 
 Department of Mathematics\\
 City University of New York, Baruch College and Graduate Center\\
 New York, NY 10010}
\thanks{The research of L.-P. A. is supported in part by NSF
Grant~DMS-1513441 and PSC-CUNY Research Award~68784-00~46.}
\email{louis-pierre.arguin@baruch.cuny.edu}

\author[C.M. Newman]{C.M. Newman}            
 \address{C.M. Newman\\ 
 Courant Institute of Mathematical Sciences\\
 New York, NY 10012 USA\\
  and NYU-ECNU Institute of Mathematical Sciences at NYU Shanghai\\
  3663 Zhongshan Road North, Shanghai 200062, China}
\thanks{The research of CMN is supported in part by U.S.~NSF Grant DMS-1207678.}
\email{newman@cims.nyu.edu}

\author[D.L.~Stein]{D.L.~Stein}            
 \address{D.L.~Stein\\ 
 Department\ of Physics and Courant Institute of Mathematical Sciences\\
  New York University\\
 	 New York, NY 10003, USA\\
	 and NYU-ECNU Institutes of Physics and Mathematical Sciences at NYU Shanghai\\
	  3663 Zhongshan Road North\\
	  Shanghai, 200062, China}
\thanks{The research of DLS is supported in part by U.S.~NSF Grant DMS-1207678.
A part of his work on this article was supported by a John Simon Guggenheim Foundation Fellowship.}
\email{daniel.stein@nyu.edu}

\author[J. Wehr]{J. Wehr}            
 \address{J. Wehr\\ 
Department of Mathematics\\
University of Arizona\\
Tucson, AZ 85721, USA
 }
\thanks{The research of JW is supported in part by U.S.~NSF Grant DMS-131271. A part of his work on this article was supported by U.S.~NSF 
Grant DMS-1440140 while he was in residence at the Mathematical Sciences Research Institute in Berkeley during the Fall 2015 semester.}
\email{wehr@math.arizona.edu}


\date{}



\maketitle

\begin{abstract}
We consider the energy difference restricted to a finite volume
for certain pairs of incongruent ground states  (if they exist) in
the $d$-dimensional Edwards-Anderson (EA) Ising spin glass at zero temperature.  We
prove that the variance of this quantity with
respect to the couplings grows at least proportionally to the volume
in any $d \ge 2$.  An essential aspect of our result is the use of the {\it excitation\/} metastate. As an illustration
of potential applications, we use this result to restrict the
possible structure of spin glass ground states in two dimensions.
\end{abstract}

\section{Introduction}
\label{sec:intro}

In a previous paper~\cite{ANSW14}, the authors considered the free energy difference
restricted to a finite volume for certain pairs of (putative) incongruent states~\cite{HF87,FH87} in the Edwards-Anderson 
Ising spin glass~\cite{EA75} at nonzero temperature, and proved that the variance of the 
free energy difference for a given pair grew linearly with the volume in any dimension
greater than or equal to two. The proof was restricted to incongruent states chosen from different metastates, or else
those chosen from the same metastate but with different nonzero weights in that metastate.

There are two directions in which the result in~\cite{ANSW14} might be generalizable. One direction is
to enlarge the class of incongruent pure state pairs to which the result applies. The second is to extend the result
to incongruent pairs at zero temperature. This paper addresses the latter.   Our main result is Theorem~\ref{thm:zerolower} below.  

While it seems intuitively reasonable that the result should remain valid at zero temperature, a number of technical issues prevent a straightforward extension 
of the proof in this direction. As will be shown below, the original proof can be thought of as comprising two logical components. 
The first component employs
natural covariance properties of the periodic boundary condition metastate~\cite{AW90,NS96,NS97,NS98}, along with translation-invariance of the coupling distribution,
to equate a suitably defined derivative of the free energy difference, with respect to a specific coupling, to the difference of the expectations of the corresponding (nearest neighbor) two-spin
correlation function in each of the incongruent pure states of the pair under study. The assumption of incongruence implies this difference will be of order one with positive probability
in the couplings.

The second component involves the construction of a martingale decomposition of the free energy difference. This involves dividing the volume into blocks whose size is
independent of the volume but large enough so that the
metastate average of the difference of the two-spin correlation function in the two incongruent states is nonzero. A sequence of
metastate averages of the free energy differences is then constructed, with
each element in the sequence conditioned on the couplings in an
increasing number of blocks. The variance of the free energy
difference over all the couplings in the volume is shown
to be no smaller than the sum of the variances of the differences
between two succeeding elements of the
sequence. The remainder of the proof uses the translation-covariant properties
of the metastate to show that the variance of the $k^{\rm
  th}$~difference is independent of $k$. Because (by the assumption of
existence of incongruence) the variance of the first block is of
order~one independent of the volume, the result follows.

In attempting to extend this result to zero temperature, a problem is encountered. Consider two ground
states distinguishable in a fixed volume $\Lambda\subset\mathbb{Z}^d$ centered at the origin. Now an issue arises in equating
the derivative of the energy difference with respect to a coupling with the two-spin expectation, because a small change in a coupling
can change the ground state if the coupling magnitude is close to its critical
value~\cite{NS00,NS01,ADNS10} for that ground state. What may appear at first as a small technical issue turns out
to require a complete revision of the first component of the proof (the second component, involving the martingale construction, then proceeds 
essentially unchanged). Indeed, a  complete treatment requires using a different construct entirely: the {\it excitation metastate\/}, introduced in~\cite{NS01} 
and developed further in~\cite{ADNS10}. In the next section we introduce the key concepts and definitions needed to
extend the result to zero temperature.

\section{Setup of the Problem}
\label{sec:setup}

As in~\cite{ANSW14}, we study the Edwards-Anderson~(EA) Hamiltonian~\cite{EA75} in a finite volume $\Lambda=[-L,L]^d\subset \Z^d$
centered at the origin:
\begin{equation}
\label{eq:EA}
H_{\Lambda, J}(\sigma)=-\sum_{(x,y)\in E(\Lambda)}J_{xy} \sigma_x\sigma_y, \qquad \sigma\in \{-1,1\}^{\Lambda} \ ,
\end{equation}
where $E(\Lambda)$ denotes the set of edges with both
endpoints in $\Lambda$.  The couplings $J_\Lambda:=(J_{xy}, (x,y)\in E(\Lambda))$
are i.i.d.~random variables sampled from a continuous distribution $\nu(dJ_{xy})$.  We assume throughout that $\nu(dJ_{xy})$
is symmetric, i.e., invariant under $J_{xy} \to -J_{xy}$,
and that  $\int \nu(dJ_{xy}) J_{xy}^4<\infty$. On every volume periodic boundary conditions are imposed, so
that ground states appear as spin-reversed pairs.

Next, let $\Sigma=\{-1,+1\}^{\Z^d}$ and let 
$\mathcal M_1(\Sigma)$ be the set of (regular Borel) probability measures on $\Sigma$. 
An infinite-volume Gibbs state~$\Gamma$ for the 
Hamiltonian~\eqref{eq:EA} is an element of $\mathcal M_1(\Sigma)$ that
satisfies the DLR equations~\cite{DLR} for that Hamiltonian
(at a given inverse temperature $\beta$). 
For a function $f(\sigma)$ on the spins, we denote 
\begin{equation}
\label{eq:Gamma}
\Gamma(f(\sigma))=\int\ d\Gamma\  f(\sigma)\, .
\end{equation}
The set of Gibbs states corresponding to the coupling realization~$J$ is denoted by $\mathcal G_J$.

At zero temperature ($\beta \to \infty$), a Gibbs state is supported on (infinite-volume)
ground states --- i.e., on spin configurations such that any change of finitely many
spin variables gives a positive energy change. An infinite-volume ground state can also be considered as the
infinite-volume limit of a convergent (sub)sequence of finite-volume ground states, defined as 
the spin configuration pairs $\tilde\sigma_\Lambda\in \{-1,+1\}^\Lambda$ that minimize $H_{\Lambda,J}$ over $\{-1,1\}^\Lambda$. 
By our assumption on the distribution of the disorder, for almost all realizations of~$J$ the ground state pair in every volume is unique.

\section{The Differentiation Lemma and Fluctuation Bound}
\label{sec:diff}

In~\cite{ANSW14} we studied the quantity
\begin{equation}
\label{eq: F}
F_\Lambda(J,\Gamma,\Gamma')=\log \frac{\Gamma(\exp \beta H_{\Lambda,J}(\sigma))}{\Gamma'(\exp \beta H_{\Lambda,J}(\sigma'))}\ ,
\end{equation}
which is the difference of free energies restricted to a finite volume at fixed $\beta<\infty$
between two infinite-volume Gibbs states $\Gamma$ and $\Gamma'$ for
the Hamiltonian~\eqref{eq:EA}. Given two probability measures, denoted $\kappa_J$ and $\kappa'_J$, on $\mathcal M_1(\Sigma)$,
we considered the free energy fluctuations of $\Gamma$ and $\Gamma'$ independently chosen under the probability measure
\begin{equation}
\label{eq: M}
M:=\nu(dJ) \ \kappa_J (d\Gamma')\times \kappa'_J(d\Gamma)\ .
\end{equation}
The most natural measures $\kappa_J$ to consider are the ones that are obtained by taking
(deterministic) subsequential limits of finite-volume Gibbs measures,
as discussed in~\cite{AW90,NS96,NS97,NS98,AD11}.  In so doing the
$\kappa_J$'s will inherit useful invariance properties from the
finite-volume Gibbs measures, as will be discussed in the next section.
For now, we simply note that these properties led to a crucial lemma
(Lemma~5.4 in~\cite{ANSW14}) that showed smoothness properties of
the conditional expectation of the free energy difference given the couplings inside a given volume. 
In particular, we can compute its derivatives, as described in the following Lemma.
To state it, we write $M_{B^c}$ for the measure defined as $M$ but with the couplings inside $B$ set to $0$, i.e.,
\begin{equation}
M_{B^c}=\nu(dJ_{B^c})\ \kappa_{J_{B^c}}(d\Gamma)\times \kappa'_{J_{B^c}}(d\Gamma').
\end{equation}
Moreover, we define the operation $L_{J_B} : \Gamma \mapsto L_{J_B}\Gamma$ where
\begin{equation}
\label{eq: gamma L}
\Big(L_{J_B} \Gamma\Big)(f(\sigma))= \frac{\Gamma\Big(f(\sigma) \exp\Bigl(-\beta H_{B,J}(\sigma)\Bigr)\Big)}{\Gamma\Big(\exp\Bigl(-\beta H_{B,J}(\sigma)\Bigr)\Big)}\, ,
\end{equation}
which simply modifies the couplings within a finite subset $B$ of $\Z^d$.
\begin{lem}
\label{lem:deriv}
(Arguin-Newman-Stein-Wehr~\cite{ANSW14}) Let $F_\Lambda$ be as in \eqref{eq: F} and $B\subset \Lambda$ finite. Then for any edge $(x,y)$ in $B$, we have
\begin{equation}
\frac{\partial}{\partial J_{xy}} M\big(F_\Lambda | J_B\big)= \beta \ M_{B^c} \big( L_{J_B}\Gamma(\sigma_x\sigma_y)-  L_{J_B}\Gamma'(\sigma_x\sigma_y)\big) \text{ $\nu$-a.s.}
\end{equation}
\end{lem}

With this lemma in hand, the desired lower bound can be obtained from a martingale decomposition as noted above; we refer the reader to~\cite{ANSW14} for details.
Here we simply state the main result of~\cite{ANSW14}, which is a lower bound on the variance of fluctuations of  $F_\Lambda$ under the measure $M$, if incongruent states
are present. That is, we assume the following is true:
\begin{ass}
\label{posass}
There exists an edge $(x,y)\in E(\Z^d)$ such that 
\begin{equation}
\label{eq: posass}
\nu\Bigl\{ J: \kappa_J\big(\Gamma(\sigma_x\sigma_y)\big)\neq \kappa'_J\big(\Gamma(\sigma_x\sigma_y)\big)\Bigr\}>0\ \,\,\,\,\,\,\,\,\,{\rm (assumption\, of\, incongruence)}\, .
\end{equation}
\end{ass}

Our main result was:
\begin{thm}
\label{thm:lower}
If Assumption \ref{posass} holds, then there exists a constant $c>0$ such that for any $\Lambda=[-L,L]^d \subset \Z^d$
the variance of $F_\Lambda$ under $M$ satisfies
\begin{equation}
\text{Var}_M\Big(F_\Lambda\Big)\geq c |\Lambda |\ .
\end{equation}
\end{thm}

Our goal here is to extend Theorem~\ref{thm:lower} to energy fluctuations between incongruent ground states at zero temperature.  This extension is carried out in Section 5---see Theorem~\ref{thm:zerolower}.

\section{Covariance Properties of the Metastate}
\label{sec:covariance}

The complete proof of Lemma~\ref{lem:deriv} appears in~\cite{ANSW14}. The essence of the proof relies heavily on the covariance properties of the metastate at nonzero temperature. Because these
will also be essential in the argument presented below, we recount them here. 

\begin{df}
\label{df: metastate}
A metastate $\kappa_{\cdot}$ for the EA Hamiltonian on $\Z^d$ is a measurable mapping 
\begin{equation}
\begin{aligned}
\R^{E(\Z^d)} &\to \mathcal M_1(\Sigma)\\
J &\mapsto \kappa_{J}
\end{aligned}
\end{equation}
with the following properties:

\medskip 

\begin{enumerate}

\item {\bf Support on Gibbs states.} Every state sampled from $\kappa_J$ is a Gibbs state for the realization for the couplings. Precisely,
\begin{equation}
\kappa_{J}\Bigl(\mathcal G_{J}\Bigr)=1 \ \text{$\nu$-a.s.}\nonumber
\end{equation}

\smallskip 

\item {\bf Coupling Covariance. } For  $B\subset \Z^d$ finite, $J_B\in \R^{E(B)}$, and any measurable subset $A$ of  $\mathcal M_1(\Sigma)$,
\begin{equation}
\kappa_{J+J_B}(A)= \kappa_{J}(L_{J_B}^{-1}A)\nonumber
\end{equation}
where $L_{J_B}^{-1} A=\Bigl\{\Gamma\in \mathcal M_1(\Sigma): L_{J_B}\Gamma \in A\Bigr\}$.

\smallskip

\item {\bf Translation Covariance.} For any translation $T$ of $\Z^d$ and any measurable subset $A$ of  $\mathcal M_1(\Sigma)$
\begin{equation}
\kappa_{TJ}(A)=\kappa_{J}(T^{-1}A)\ .\nonumber
\end{equation}
\end{enumerate}
\end{df}

The coupling covariance is useful in making explicit the dependence of the metastate on $J_B$. This is crucial because the differentiation in Lemma~\ref{lem:deriv} has to
take into account the dependence of $F_\Lambda$ on all of the couplings. We therefore need a zero-temperature object that possesses properties similar to the ordinary metastate.
One cannot simply generalize to zero temperature, however, because the useful property of coupling covariance no longer holds. What is needed is
a probability measure on ground states that keeps track of how ground states change under local coupling modifications. This quantity is the excitation metastate,
which we now  describe.

\section{Excitation Metastate}
\label{sec:excitation}

In this section we define the excitation metastate~\cite{NS01,ADNS10} and study its covariance properties under translations and changes in couplings.
We begin by defining the needed measure in a finite volume $\Lambda$, and then define the limiting object as $\Lambda\to\Z^d$. 
Fix a box $B\subset \Lambda$. The ground state configuration in $B$ can be chosen as follows. Let $\sigma^\eta$ denote the minimizer of $H_{\Lambda,J}$ over the spin configurations of $\Lambda$ that are equal to $\eta$ on $B$.
Then the spins of the ground state in $B$ can be determined among all $\eta\in\{-1,+1\}^B$  as follows:
\begin{equation}
\text{$\sigma|_B$ is the unique $\eta\in\{-1,1\}^B$ such that $H_{\Lambda,J}(\sigma^\eta)-H_{\Lambda,J}(\sigma^{\eta'})<0$ for all $\eta'\neq \eta$.}
\end{equation}
Now let $J_B$ be a configuration of the couplings that is $0$ for edges outside $B$. Write $\sigma(J_B)$ for the ground state in $\Lambda$ for the disorder $J+J_B$. 
Again, the value of the spins in $B$ for the ground state $\sigma(J_B)$ can be determined as follows:

\medskip

$\sigma(J_B)\Big|_B$ is the unique $\eta\in\{-1,1\}^B$ such that
\begin{equation}
H_{\Lambda,J}(\sigma^\eta)-H_{\Lambda,J}(\sigma^{\eta'})+H_{B,J_B}(\eta)-H_{B,J_B}(\eta')<0 \text{ for all $\eta'\neq \eta$\, ,}
\end{equation}
which follows because
\begin{equation}
H_{\Lambda,J+J_B}(\sigma^\eta)=H_{\Lambda,J}(\sigma^\eta)+H_{B,J_B}(\eta)\, .
\end{equation}
For conciseness, write
\begin{equation}
\Delta E_B(\eta, \eta')=H_{\Lambda,J}(\sigma^\eta)-H_{\Lambda,J}(\sigma^{\eta'})
\end{equation}
An elementary decoupling argument shows that the difference of energy is bounded uniformly in $\Lambda$:
\begin{equation}
|\Delta E_B(\eta, \eta')|\leq |H_{J,B}(\eta)-H_{J,B}(\eta')|\, .
\end{equation}
In particular, the above shows that the sequence of the distributions of $\Delta E_B(\eta, \eta')$ (induced by that of $J$) over the sequence of volumes is tight.
These observations lead to the existence of a limiting measure.

\begin{lem}
Fix $B\subset \Z^d$ finite. 
For every $\Lambda\subset \Z^d$ consider the joint distribution of  $\big(J,\vec{\sigma}_B, \Delta\vec{E}_B\big)$ where $\vec{\sigma}_B=\Big(\sigma^{\eta}, \eta\in\{-1,1\}^B\Big)$ and $\Delta\vec{E}_B=\Big(\Delta E_B(\eta, \eta'); \eta,\eta'\in\{-1,1\}^B\Big)$ as constructed above.
There exists a subsequence of volumes such that the distributions converge weakly. 
\end{lem}
We write $\kappa_J$ for the limiting conditional distribution on $(\vec{\sigma}_B, \Delta\vec{E}_B)$ given $J$. 
The subsequence of volumes can be picked so that the convergence holds jointly for all finite subsets $B\subset \Z^d$. 
For simplicity, we will restrict ourselves to a fixed $B$. 
We refer to the distribution $\kappa_J$ as an excitation metastate. Formally:
\begin{df}
\label{df: metastate}
An excitation metastate $\kappa_{\cdot}$ for the EA Hamiltonian on $\Z^d$ is a measurable mapping from $J$ to the set of probability measures
on $(\vec{\sigma}_B,\Delta\vec{E}_B)$ where $\vec{\sigma}_B=\Big(\sigma^{\eta}, \eta\in\{-1,1\}^B\Big)$ and $\Delta\vec{E}_B=\Big(\Delta E_B(\eta, \eta');\eta,\eta'\in\{-1,1\}^B\Big)$, with the following properties:

\medskip

\begin{enumerate}
\item {\bf Support on Ground States} 
Denote by $\sigma$ the unique configuration in $\Big(\sigma^{\eta}, \eta\in\{-1,1\}^B\Big)$ such that
\begin{equation}
\Delta E_B(\eta, \eta')<0 \qquad \text{ for every $\eta'\neq \eta$.}
\end{equation}
Then $\sigma$ is a ground state for the EA Hamiltonian on $\Z^d$ with disorder $J$.

\medskip

\item {\bf Coupling Covariance} 
Let $J_B$ be a a coupling that is $0$ for every edge not in $E(B)$. Denote by $\sigma(J_B)$ the unique configuration in $\{-1,1\}^{\Z^d}$ with the property that the configuration of spins in $B$ is given by the unique $\eta\in \{-1,1\}^B$ such that
\begin{equation}
\Delta E_B(\eta, \eta')+H_{B,J_B}(\eta)-H_{B,J_B}(\eta')<0 \qquad \text{ for every $\eta'\neq \eta$.}
\end{equation}
Then for any measurable function $F$ of the ground state $\sigma$
\begin{equation}
\int_{\{-1,1\}^{\Z^d}} F(\sigma) \ \kappa_{J+J_B}(d\sigma) =\int_{\{-1,1\}^{\Z^d}} F(\sigma(J_B)) \ \kappa_J(d\sigma)\, .
\end{equation}
Here we abuse  notation and write $\kappa_J(d\sigma)$ to denote the marginal distribution of $\sigma$ induced by $\kappa_J$. 

\medskip

\item {\bf Translation Covariance.} For any translation $T$ of $\Z^d$ and any measurable subset $A$ of  $\{-1,+1\}^{\Z^d}$
\begin{equation}
\kappa_{TJ}(A)=\kappa_{J}(T^{-1}A)\ .
\end{equation}
This is a direct consequence of the periodic boundary conditions. 
\end{enumerate}
\end{df}

We are now ready to prove the differentiation lemma. Consider the energy difference between two ground states in $\Lambda$.
Let
\begin{equation}
G_\Lambda(J,\sigma,\sigma')=H_{\Lambda, J}(\sigma)-H_{\Lambda, J}(\sigma').
\end{equation}
We consider this difference as a random variable when $\sigma$ and $\sigma'$ are two ground states sampled from two excitation metastates for the disorder $J=(J_{B^c},J_B)$; that is,
\begin{equation}
M=\nu(dJ)\ \kappa_J(d\sigma)\times \kappa'_J(d\sigma')\ .
\end{equation}
We write $M_{B^c}$ for the corresponding measure with the couplings inside $B$ set to $0$, i.e.,
\begin{equation}
M_{B^c}=\nu(dJ_{B^c})\ \kappa_{J_{B^c}}(d\sigma)\times \kappa'_{J_{B^c}}(d\sigma').
\end{equation}
\begin{lem}
\label{lem: deriv}
Recall the definition of $\sigma(J_B)$ in the coupling covariance part of Definition~\ref{df: metastate}. Then for every edge $(x,y)$ in $E(B)$,
\begin{equation}
\frac{\partial}{\partial J_{xy}} M\big(G_\Lambda(J,\sigma,\sigma')| J_B\big)=-M_{B^c}\big(\sigma_x(J_B)\sigma_y(J_B)-\sigma'_x(J_B)\sigma'_y(J_B)\big)
\text{ $J_B$-a.s.}
\end{equation}
\end{lem}
The interpretation of the right hand side is straightforward. By the coupling covariance property, it has the same distribution as
\begin{equation}
M(\sigma_x\sigma_y-\sigma'_x\sigma'_y| J_B)\ .
\end{equation}
That is, the right hand side in the statement of the lemma in effect sets the coupling inside $B$ to $0$ but puts it back explicitly with the help of the mapping $J_B$. 
As we will show in the proof of Theorem~\ref{thm:zerolower}, by a martingale argument similar to that used in~\cite{ANSW14} for the positive temperature case, this quantity is strictly positive for $B$ large enough.
\medskip

\begin{proof}[\bf Proof of Lemma 5.3] Using coupling covariance, the left-hand side is the same as
\begin{equation}
\frac{\partial}{\partial J_{xy}} M\big(G_\Lambda(J,\sigma,\sigma')| J_B\big)
=\frac{\partial}{\partial J_{xy}} M_{B^c}\big(G_\Lambda(J,\sigma(J_B),\sigma'(J_B))\big).
\end{equation}
Without loss of generality, we can work with a single $H_{\Lambda,J}(\sigma(J_B))$ since $G_\Lambda$ is simply the difference of the two ground state energies.

The crucial observation is the following. Let $\varepsilon_{xy}$ be a small variation of the coupling for the edge $(x,y)$. We can write the derivative of the ground state energy as
\begin{equation}
\lim_{\varepsilon_{xy}\to 0} \frac{H_{\Lambda,J+\varepsilon_{xy}}(\sigma(J_B+\varepsilon_{xy}))-H_{\Lambda,J}(\sigma(J_B))}{\varepsilon_{xy}}\, .
\end{equation}
To ensure that we can interchange the derivative with the expectation $M_{B^c}$, we have to ensure that the ratio stays bounded as $\varepsilon_{xy}\to 0$. 
To do this, we use the fact that
\begin{equation}
\begin{aligned}
H_{\Lambda,J+\varepsilon_{xy}}(\sigma(J_B+\varepsilon_{xy}))&\leq H_{\Lambda,J+\varepsilon_{xy}}(\sigma(J_B))\\
H_{\Lambda,J}(\sigma(J_B))&\leq H_{\Lambda,J}(\sigma(J_B+\varepsilon_{xy}))\\
\end{aligned}
\end{equation}
by the definition of the ground state energy. We therefore obtain the bound
\begin{equation}
-\varepsilon_{xy}\sigma_x(J_B)\sigma_y(J_B)\leq H_{\Lambda,J+\varepsilon_{xy}}(\sigma(J_B+\varepsilon_{xy}))-H_{\Lambda,J}(\sigma(J_B))\leq -\varepsilon_{xy}\sigma_x(J_B+\varepsilon_{xy})\sigma_y(J_B+\varepsilon_{xy})
\end{equation}
This demonstrates two things: first, and most importantly, the ratio stays bounded as $\varepsilon_{xy}\to 0$; second, the derivative is 
\begin{equation}
\frac{\partial }{\partial J_{xy}}H_{\Lambda,J}(\sigma(J_B))=-\sigma_x(J_B)\sigma_y(J_B)\ .
\end{equation}
Here we used the fact that for almost all $J_B$, the ground state $\sigma(J_B)$ is locally constant by the coupling covariance property (recall that the distribution of $J_{xy}$ is continuous).
Finally, by the dominated convergence theorem, we have 
\begin{equation}
\frac{\partial}{\partial J_{xy}} M\big(G_\Lambda(J,\sigma,\sigma')| J_B\big)=-\sigma_x(J_B)\sigma_y(J_B)\, ,
\end{equation}
which proves the lemma.
\end{proof}
To state the main result of the paper, we need the following analog of Assumption~\ref{posass}:
\begin{ass}
\label{ass:zeroass}
There exists an edge $(x,y)\in E(\Z^d)$ such that 
\begin{equation}
\label{eq: zeroass}
\nu\Bigl\{ J: \kappa_J(\sigma_x\sigma_y)\neq \kappa'_J(\sigma_x\sigma_y)\Bigr\}>0\ \,\,\,\,\,\,\,\,\,{\rm (assumption\, of\, incongruence)}\, .
\end{equation}
\end{ass}
Our main result is as follows:
\begin{thm}
\label{thm:zerolower}
If Assumption \ref{ass:zeroass} holds, then there exists a constant $c>0$ such that for any $\Lambda=[-L,L]^d \subset \Z^d$
the variance of $G_\Lambda$ under $M$ satisfies
\begin{equation}
\text{Var}_M\Big(G_\Lambda\Big)\geq c |\Lambda |\ .
\end{equation}
\end{thm}
\begin{proof}[\bf Proof of Theorem~\ref{thm:zerolower} (Sketch)]
In essence, until now we have been preparing for the proof.  The idea is to repeat the proof of the main result of \cite{ANSW14}, replacing the quantity $F_{\Lambda}$ used there by $G_{\Lambda}$ introduced above.  The pattern of the proof in \cite{ANSW14} is followed very closely, so we  present only a sketch.  As in \cite{ANSW14},  we divide $\Lambda$ into blocks  $B_1,\dots, B_N$ of equal size.  The size of the blocks is chosen  independently of $\Lambda$, so as to make the contribution from each of the blocks positive (and equal).  It follows that the number of the blocks $N$ is of the order of the size of $\Lambda$:  $N = C|\Lambda|$.  To carry out the estimate, we introduce the $\F_k$ generated by the couplings $J_{B_i}$, $i\leq k$,  and use a martingale decomposition to obtain:
\begin{equation}
\label{eq:martdiff}
\text{Var}_M ( G_\Lambda)\geq \sum_{k=1}^N \text{Var}_M\Bigl(M( G_\Lambda| \F_k) - M( G_\Lambda| \F_{k-1})\Bigr)\, .
\end{equation}
To show that  block size can be chosen to yield  nonzero contributions from the terms on the right-hand side, we use Lemma 5.3, together with Assumption 3.2.  As in \cite{ANSW14}, we can use the Martingale Convergence Theorem to choose the size of the blocks, so that the derivative of the martingale difference $M(G_\Lambda| \F_k)- M( G_\Lambda| \F_{k-1})$ with respect to the coupling at the center of the block is a nonzero random variable.  It follows that each variance on the right-hand side of the above inequality is a nonzero constant, so that their sum is proportional to $|\Lambda|$, as claimed.
\end{proof}

\section{Some Remarks on Ground State Scenarios in Two Dimensions}
\label{sec:2D}

In any dimension, an almost sure upper bound on $G_{\Lambda}$ always holds by decoupling the boundary:
\begin{equation}
\Bigl| G_\Lambda(J,\sigma,\sigma')\Bigr| \leq 4 \sum_{e\in \partial \Lambda} |J_e|\ , \text{ $M$-a.s.}
\end{equation}
In addition, a lower bound on the moment generating function of $G_{\Lambda}$ can be proved as in \cite{ANSW14}, using an argument similar to the proof of the martingale central limit theorem.
This yields
\begin{equation}
\label{eq: lower thm}
\int\nu(dJ_\Lambda)\  \exp \left(t\ \frac{M( G_\Lambda| J_\Lambda)}{|\partial \Lambda|}\right) \geq e^{c t^2}\ .
\end{equation}
As in \cite{ANSW14}, the two bounds lead to a contradiction for some {\it a priori\/} scenarios for the excitation metastate structure in two dimensions. 
In particular, suppose a metastate~$\kappa$ is supported on a countable (finite or infinite) set of ground state pairs $(\sigma^{(n)}, -\sigma^{(n)})$, with $-\sigma$ denoting the global flip of~$\sigma$.  Although the $\sigma^{(n)}$ depend on the coupling realization, by a straightforward application of the ergodic theorem the $\kappa$-measure of the $n^{\rm th}$ pair is an a.s.~constant number $p_n$.  It follows from our results that $p_m$ cannot be different from $p_n$ for any $m$ and $n$, by the same reasoning that led to a similar conclusion for positive-temperature pure state pairs in~\cite{ANSW14} (cf.~Sect.~7).

It follows that no metastate supported on a countable infinity of ground state pairs is possible---in such a metastate some $p_n$ would have to be distinct, since $\sum_n p_n = 1$.  As for metastates supported on a finite number of ground state pairs, the only ones that are not eliminated by the results presented here are those with $p_n = {1 \over M}$ for a finite~$M$.  We finish with an open problem:  can one prove that such a symmetric metastate cannot exist for $M > 1$ and, as a consequence, the number of ground state pairs in the support of the metastate is either one or (an uncountable) infinity?  Analogous results are known for random ferromagnets in any dimension~\cite{WW16} and for the Ising spin glass on a half-plane~\cite{ADNS10,AD11}.

\end{document}